\documentclass[letterpaper, 10 pt, conference]{ieeeconf}
\IEEEoverridecommandlockouts
\usepackage{graphicx}
\usepackage{amsmath,amssymb}
\usepackage{flushend}
\usepackage{color}
\usepackage{latexsym}\usepackage{comment}
\usepackage{url}
\usepackage{graphics}
\usepackage{pdfsync}
\usepackage{stfloats}
\usepackage{epsfig}
\usepackage{times}
\usepackage{amssymb}
\usepackage{amsmath}
\usepackage{amsfonts}
\usepackage{textcomp}
\usepackage{xcolor}
\usepackage{indentfirst}
\usepackage{epstopdf}
\usepackage{multirow}
\usepackage{float}
\usepackage{color}
\usepackage{enumerate}
\usepackage{subfigure}
\usepackage{comment}
\usepackage{cite}
\usepackage{caption}
\usepackage{bm}

\usepackage[ruled, vlined]{algorithm2e}
\usepackage{algorithmic}

\DeclareMathOperator*{\minimize}{minimize}

\newcommand{\ba}{\begin{array}}
	\newcommand{\ea}{\end{array}}
\newcommand{\be}{\begin{equation}}
\newcommand{\ee}{\end{equation}}
\newcommand{\bea}{\begin{eqnarray}}
\newcommand{\eea}{\end{eqnarray}}
\newcommand{\bean}{\begin{eqnarray*}}
	\newcommand{\eean}{\end{eqnarray*}}
\newcommand{\bc}{\begin{center}}
	\newcommand{\ec}{\end{center}}

\title{\LARGE \bf Data-Driven Min-Max MPC for Linear Systems}
\author{Yifan~Xie, Julian~Berberich, Frank Allg\"{o}wer
	\thanks{F. Allg\"{o}wer is thankful that his work was funded by Deutsche Forschungsgemeinschaft (DFG, German Research
Foundation) under Germany’s Excellence Strategy - EXC 2075 - 390740016 and under grant 468094890.
F. Allg\"{o}wer acknowledges the support by the Stuttgart
Center for Simulation Science (SimTech).
The authors thank the International Max Planck Research School
for Intelligent Systems (IMPRS-IS) for supporting Yifan Xie.
	}
	\thanks{Yifan~Xie, Julian~Berberich, Frank Allg\"{o}wer are with the Institute for Systems Theory and Automatic Control, University of Stuttgart, 70550 Stuttgart, Germany.
		{\tt\small  \{yifan.xie, julian.berberich, frank.allgower\}@ist.uni-stuttgart.de}. }
}

\newtheorem{mythm}{Theorem}

\newtheorem{mylem}{Lemma}

\newtheorem{remark}{Remark}
\newtheorem{assum}{Assumption}
\IncMargin{0.8em}

\begin{document}
	
	\maketitle

\begin{abstract}
Designing data-driven controllers in the presence of noise is an important research problem, in particular when guarantees on stability, robustness, and constraint satisfaction are desired.
In this paper, we propose a data-driven min-max model predictive control (MPC) scheme to design state-feedback controllers from noisy data for unknown linear time-invariant (LTI) system.
The considered min-max problem minimizes the worst-case cost over the set of system matrices consistent with the data.
We show that the resulting optimization problem can be reformulated as a semidefinite program (SDP).
By solving the SDP, we obtain a state-feedback control law that stabilizes the closed-loop system and guarantees input and state constraint satisfaction.
A numerical example demonstrates the validity of our theoretical results.
\end{abstract}

\section{Introduction}\label{sec:1}
Model Predictive Control (MPC) is an advanced control technique that can handle nonlinear multi-input multi-output systems with constraints \cite{rawlings2017model}.
The basic idea of MPC is to solve an open-loop optimal control problem at each sampling time, which uses the current state as the initial condition and system dynamics to predict future open-loop states.

Ensuring robust constraint satisfaction in the presence of noise or model uncertainty is challenging.
Several MPC methods have been proposed to deal with this issue.
Tube-based robust MPC \cite{mayne2005robust,mayne2011tube}, where a constraint tightening is included in the MPC optimization problem, ensures that all possible realizations of the state trajectory lie in an uncertainty tube around a nominal system.
This approach typically assumes that the nominal system is known, whereas the noise is unknown and bounded.
In \cite{lorenzen2019update} and \cite{lu2021robust}, model uncertainty is characterized within a set, and additional measurements are leveraged to reduce model uncertainty and enhance the performance of the tube-based MPC scheme.
Furthermore, min-max MPC can effectively address scenarios involving model uncertainty, as discussed in \cite{kothare1996robust,bemporad2003minmax}, and \cite{scokaert1998minmaxfeedback}.
The goal of min-max MPC is to design control inputs that minimize the worst-case cost w.r.t. disturbances or uncertainties.
State-feedback control laws are commonly employed in the min-max MPC framework to decrease the computational complexity \cite{kothare1996robust, bemporad2003minmax, scokaert1998minmaxfeedback}.

Standard model-based control techniques, in particular standard MPC schemes, rely on a priori knowledge of system models that are either identified from measured data by system identification methods\cite{ljung1999system} or derived based on first principles.
In contrast, data-driven control approaches can design controllers directly from the available data.
The Fundamental Lemma proposed by
Willems et al. \cite{willems2005note} has led to control strategies in a behavioral setting \cite{markovsky2021review}, and more recently in the context of MPC \cite{coulson2019data,berberich2021guarantees}.
For example, \cite{berberich2021guarantees} proposes a robust data-driven MPC scheme and proves theoretical guarantees in case of noisy data.
The behavioral framework requires persistently exciting data, enabling the unique representation of the system from the data in the noise-free scenarios.
On the other hand, it has been shown in the informativity framework that the data need not to be sufficiently informative to uniquely identify the system \cite{van2020data}.
In this framework, control strategies are proposed to stabilize the system based on a representation of the set of system matrices consistent with the data \cite{van2020data, van2020noisy, van2023informativity, persis2020formulas}.
For example, \cite{van2020noisy} proposes a controller design method directly from noisy input-state data.
However, the design of MPC schemes, known for their effectiveness in handling constraints, remains unaddressed in this framework.

In this paper, we present a data-driven min-max MPC scheme which uses noisy data to control linear time-invariant (LTI) systems with unknown system matrices.
Our approach relies on a representation of the system matrices consistent with a sequence of noisy input-state data that was employed, e.g., by \cite{berberich2023combining,bisoffi2021trade,van2020noisy}.
We reformulate the data-driven min-max MPC problem with input and state constraints to a semidefinite program (SDP) that yields a state-feedback control law.
Further, we show that the proposed data-driven min-max MPC guarantees closed-loop properties including recursive feasibility, constraint satisfaction and exponential stability.

The remainder of this paper is organized as follows.
In Section~\ref{sec:2}, we introduce necessary preliminaries.
In Section~\ref{sec:3}, we propose the data-driven min-max MPC scheme and show that the scheme ensures recursive feasibility and exponential stability.
We apply the developed scheme to a numerical example in Section~\ref{sec:4}.
Finally, we conclude the paper in Section~\ref{sec:5}.
	
\section{Preliminaries}\label{sec:2}
Let $\mathbb{I}_{[a, b]}$ denote the set of integers in the interval $[a, b]$ and  $\mathbb{I}_{\geq 0}$ denote the set of nonnegative integers.
For a matrix $P$, we write $P\succ 0$ if $P$ is positive definite and $P\succeq 0$ if $P$ is positive semi-definite.
For a vector $x$ and a matrix $P\succ 0$, we write $\|x\|_P=\sqrt{x^\top Px}$.
For matrices $A$ and $B$ of compatible dimensions, we abbreviate $ABA^\top$ to $AB\begin{bmatrix}\star\end{bmatrix}^\top$.

\subsection{System representation}\label{sec:2.1}
We consider an unknown discrete-time LTI system
\begin{equation}\label{system}
x_{t+1}=A_sx_t+B_su_t+\omega_t,
\end{equation}
where $x_t\in\mathbb{R}^n$ denotes the state, $u_t\in\mathbb{R}^m$ denotes the input, and $\omega_t\in\mathbb{R}^n$ denotes the unknown noise for $t\in\mathbb{N}$.
The matrices $A_s\in\mathbb{R}^{n\times n}$ and $B_s\in\mathbb{R}^{n\times m}$ are unknown.
We define a sequence of input, noise and corresponding state measurements, which is denoted by
\begin{equation}\label{trajectory}\nonumber
\begin{aligned}
U_{-}&:=\begin{bmatrix}u_0 & u_1 &\ldots &u_{T-1}\end{bmatrix},\\
W_{-}&:=\begin{bmatrix}\omega_0 &\omega_1 &\ldots &\omega_{T-1}\end{bmatrix},\\
X&:=\begin{bmatrix}x_0 & x_1 &\ldots &x_T\end{bmatrix}.
\end{aligned}
\end{equation}
Throughout this paper, we assume that data of the form $X$ and $U_{-}$ are available, whereas $W_{-}$ is unknown.
The noise $\omega_t$ should satisfy the following instantaneous constraint.
\begin{assum}\label{assumption1}
For all $t\in \mathbb{N}$, the noise $\omega_t\in\mathbb{R}^n$ satisfies $\|\omega_t\|_2^2\leq\epsilon$
for a known bound $\epsilon\geq 0$.
\end{assum}

We define the set of system matrices $(A, B)$ consistent with the data $x_i, u_i, x_{i+1}, i\in\mathbb{N}$ by
\begin{equation}\nonumber
\Sigma_i:=\left\{\!(A, B):
\eqref{system} \text{ holds for some }\omega_i
\text{ satisfying } \|\omega_i\|_2^2\leq \epsilon
\right\}.
\end{equation}
This set includes all system matrices for which there exists a noise satisfying Assumption \ref{assumption1} and the system dynamics \eqref{system}.
We proceed as in  \cite{van2020noisy, bisoffi2021trade, berberich2023combining} to derive a data-driven parametrization of the system matrices.
The system dynamics \eqref{system} can be rewritten as
\begin{equation}\label{omega}\nonumber
\omega_i=x_{i+1}-A_sx_i-B_su_i,
\end{equation}
with which the set of system matrices $\Sigma_i$ can be equivalently characterized by the following quadratic matrix inequality
\begin{equation}\label{sigma}\nonumber
\Sigma_i=\left\{\!\!(A, B)\!:\!\begin{bmatrix}
I \!\!&A \!\!&B
\end{bmatrix}\!\!
\begin{bmatrix}
I &x_{i+1}\\
0 &-x_i\\
0 &-u_i
\end{bmatrix}
\!
\begin{bmatrix}
\epsilon I &0\\
0 &-I
\end{bmatrix}
\!
\begin{bmatrix}
\star
\end{bmatrix}^\top\!\succeq\! 0\!\right\}.
\end{equation}
The set of system matrices consistent with the sequence of input-state measurements $( U_-,X)$ is defined by
\begin{equation}\nonumber
\mathcal{C}=\bigcap_{i=0}^{T-1}\Sigma_i.
\end{equation}
We can characterize $\mathcal{C}$ by the following quadratic matrix inequality \cite{berberich2023combining, bisoffi2021trade}
\begin{equation}\label{C}
\mathcal{C}=\left\{\!\!(A, B)\!:\!\!\!
\begin{gathered}
\begin{bmatrix}
I \!\!&A \!\!&B
\end{bmatrix}
\!\Pi(\tau)\!
\begin{bmatrix}
I \!\!&A \!\!&B
\end{bmatrix}^\top\!\succeq\! 0,  \\
\forall \tau\!=\!(\tau_0, \ldots, \tau_{T-1}), \tau_i\!\geq\! 0, i\in\mathbb{I}_{[0, T-1]}
\end{gathered}\!
\right\}\!,
\end{equation}
where
\begin{equation}\nonumber
\Pi(\tau)=\sum_{i=0}^{T-1}\tau_i
\begin{bmatrix}
I &x_{i+1}\\
0 &-x_i\\
0 &-u_i
\end{bmatrix}
\!\!
\begin{bmatrix}
\epsilon I &0\\
0 &-I
\end{bmatrix}
\begin{bmatrix}
I &x_{i+1}\\
0 &-x_i\\
0 &-u_i
\end{bmatrix}^\top.
\end{equation}

\subsection{Problem setup}
In this paper, we employ data-driven min-max MPC for the unknown system $x_{t+1}=A_sx_t+B_su_t$ to stabilize the origin, while the input and state satisfy given constraints.
As explained in Section~\ref{sec:2.1}, the offline measurements $(U_{-}, X)$ are affected by noise satisfying the instantaneous constraint in Assumption \ref{assumption1}.
On the other hand, the data collected online during closed-loop operation are assumed to be noise-free.
This is assumed for simplicity, to avoid an additional maximization w.r.t. the noise in the min-max MPC problem.
Extending the proposed framework to handling online noise is an interesting issue for future research.
In Section \ref{sec:5}, we show with a numerical example that the proposed approach produces reliable results also in the presence of online noise.

We consider ellipsoidal constraints on the input and the state, i.e.,
\begin{equation}\nonumber
        \|u_t\|_{S_u}\leq 1, \|x_t\|_{S_x}\leq 1, \forall t\in \mathbb{N},
\end{equation}
where $S_u\succ 0$ and  $S_x\succeq 0$.
The centers of the ellipsoids for the input and state constraints are at the origin, but our results can be adapted for non-zero centers.
In order to stabilize the origin, we define the quadratic stage cost function
\begin{equation}\nonumber
l(u, x)=\|u\|_R^2+\|x\|_Q^2,
\end{equation}
where $R, Q\succ 0$.
The following results can be adapted for non-zero equilibria $(u^s, x^s)\neq (0, 0)$.

\section{ Data-Driven Min-Max MPC}\label{sec:3}
In Section~\ref{sec:3.1}, we define a general data-driven min-max MPC problem with input and state constraints.
In Section~\ref{sec:3.2}, we restrict the optimization to state-feedback control laws, which allows to reformulate the data-driven min-max MPC problem as an SDP.
The state-feedback control law at each time step can be obtained from a receding-horizon algorithm, which is proposed in Section~\ref{sec:3.3}.
Finally, we prove recursive feasibility, constraint satisfaction and exponential stability for the closed-loop system in Section~\ref{sec:3.4}.

\subsection{Min-max MPC problem}\label{sec:3.1}
At time $t$, given an initial state $x_t$, the data-driven min-max MPC optimization problem is formulated as follows:
\vspace{-10pt}
\begin{subequations}\label{mpc:nominal}
\begin{align}
J_\infty^*(x_t):=&\min_{\bar{u}(t)}\max_{(A, B)\in\mathcal{C}}\sum_{k=0}^{\infty}l(\bar{u}_k(t), \bar{x}_k(t))\label{mpc:nominal_obj}\\
\text{s.t.}\quad &\bar{x}_{k+1}(t)=A\bar{x}_k(t)+B\bar{u}_k(t),\label{mpc:nominal_con1}\\
&\bar{x}_0(t)=x_t,\label{mpc:nominal_con2}\\
&\|\bar{u}_k(t)\|_{S_u}\leq 1,\forall t\in \mathbb{N},\label{mpc:nominal_con3}\\
&\|\bar{x}_k(t)\|_{S_x}\leq 1, \forall (A, B)\in\mathcal{C},t\in \mathbb{N}.\label{mpc:nominal_con4}
\end{align}
\end{subequations}
The objective function is a minimization of the worst-case cost over all consistent system matrices in $\mathcal{C}$ by adapting the control input $\bar{u}_k(t), \forall k\in\mathbb{N}$.
In the optimization problem, $\bar{x}_{k}(t)$ and $\bar{u}_k(t)$ are the predicted state and control input at time $t+k$ based on the measurement at time $t$.
The prediction model employs the system matrices consistent with the data trajectory in constraint \eqref{mpc:nominal_con1}.
In constraint \eqref{mpc:nominal_con2}, we initialize $\bar{x}_0(t)$ as the state measurement at time $t$.
We consider that the input and state should lie in the ellipsoidal constraints in \eqref{mpc:nominal_con3} and \eqref{mpc:nominal_con4}.
The state constraint \eqref{mpc:nominal_con4} must be satisfied for any states predicted using any system matrices in $\mathcal{C}$.
In order to effectively address this problem and obtain a tractable solution, we limit our focus to a state-feedback control law of the form $u_t=F_tx_t$, where $F_t\in\mathbb{R}^{m\times n}$.

\subsection{Reformulation as an SDP}\label{sec:3.2}
In this subsection, we first prove that the state-feedback gain $F_t$ that minimizes the upper bound on the optimal cost of the min-max MPC problem \eqref{mpc:nominal_noconstraint} can be obtained by solving the SDP \eqref{sdp:nominal_no}.
Then, we reformulate the input and state constraints as linear matrix inequalities (LMIs).

We first neglect the input and state constraints \eqref{mpc:nominal_con3} and \eqref{mpc:nominal_con4} and obtain the state-feedback gain for the following data-driven min-max MPC problem
\begin{subequations}\label{mpc:nominal_noconstraint}
\begin{align}
J_\infty^*(x_t):=&\min_{\bar{u}(t)}\max_{(A, B)\in\mathcal{C}}\sum_{k=0}^{\infty}l(\bar{u}_k(t), \bar{x}_k(t))\label{mpc:nominal_noconstraint_obj}\\
\text{s.t.}\quad &\bar{x}_{k+1}(t)=A\bar{x}_k(t)+B\bar{u}_k(t),\label{mpc:nominal_noconstraint_con1}\\
&\bar{x}_0(t)=x_t.\label{mpc:nominal_noconstraint_con2}
\end{align}
\end{subequations}
The method to reformulate the input and state constraints \eqref{mpc:nominal_con3}-\eqref{mpc:nominal_con4} will be proposed later.

Our goal is to derive an upper bound on the worst-case cost over the set $\mathcal{C}$ and then to find a state-feedback control law to minimize this upper bound.
In order to derive the upper bound of the worst-case cost over all system matrices in $\mathcal{C}$, we define a quadratic function $V(x)=x^\top Px$ for $x\in \mathbb{R}^n$, where $P\succ 0$.
Suppose $V$ satisfies the following inequality for all states and inputs $\bar{x}_k(t), \bar{u}_k(t), k\in\mathbb{N}$ predicted by the system dynamics \eqref{mpc:nominal_noconstraint_con1} with any $(A, B)\in\mathcal{C}$
\begin{equation}\label{Vconstraint}
V(\bar{x}_{k+1}(t))-V(\bar{x}_k(t))\leq -l(\bar{u}_k(t), \bar{x}_k(t)).
\end{equation}
To ensure that the cost in equation \eqref{mpc:nominal_noconstraint_obj} is finite, we must have $\lim\limits_{k\to\infty}\bar{x}_k(t)=0$.
Therefore, we have $\lim\limits_{k\to\infty}V(\bar{x}_k(t))=0$.
Summing the inequality \eqref{Vconstraint} from $k=0$ to $k=T$ along an arbitrary trajectory and letting $T\to \infty$, we obtain
\begin{equation}\label{Vsum}
-V(\bar{x}_0(t))\leq -\sum_{k=0}^{\infty}l(\bar{u}_k(t), \bar{x}_k(t)).
\end{equation}
Since $x_0(t)=x_t$ and the inequality \eqref{Vsum} holds for any matrices $(A, B)\in\mathcal{C}$, it also holds for the worst-case value, i.e.,
\begin{equation}\label{worstV}
    \max_{(A, B)\in\mathcal{C}}\sum_{k=0}^{\infty}l(\bar{u}_k(t), \bar{x}_k(t))\leq V(x_t).
\end{equation}
This provides an upper bound on the cost \eqref{mpc:nominal_noconstraint_obj}.
The above method to derive the upper bound on the worst-case cost, i.e., \eqref{Vconstraint}-\eqref{worstV}, is inspired by the existing LMI-based min-max MPC approach in \cite{kothare1996robust}.

The goal of our data-driven min-max MPC problem is to synthesize a state-feedback control law $u_t=F_tx_t$  to minimize the upper bound $V(x_t)$ satisfying \eqref{Vconstraint} for any $(A, B)\in\mathcal{C}$.
As the following theorem shows, this is possible based on LMIs.

\begin{mythm}\label{theorem1}
Suppose that there exist $\gamma>0$, $H\in\mathbb{R}^{n\times n}$, $ L\in\mathbb{R}^{m\times n}$, $\tau\in\mathbb{R}^T$ such that the inequalities \eqref{sdp:nominal_no} hold
\begin{subequations}\label{sdp:nominal_no}
\begin{align}
    & \begin{bmatrix}1 &x_t^\top\\
x_t &H\end{bmatrix}\succeq 0, \label{sdp:nominal_no_con1}\\
    &\begin{bmatrix}
        \begin{bmatrix}
            -H &0\\
            0 &0
        \end{bmatrix}+\Pi(\tau) &
        \begin{bmatrix}
            0\\
            H\\
            L
        \end{bmatrix}
        & 0\\
        \begin{bmatrix}
            0 &H &L^\top
        \end{bmatrix} &-H &\Phi^\top\\
        0 &\Phi &-\gamma I
    \end{bmatrix}\prec 0, \label{sdp:nominal_no_con2}\\
    &\tau=(\tau_0, \ldots, \tau_{T-1}), \tau_i\geq 0, \forall i\in\mathbb{I}_{[0, T-1]},\label{sdp:nominal_no_con3}
\end{align}
\end{subequations}
where $\Phi=\begin{bmatrix}M_R L\\ M_Q H\end{bmatrix}$ and $M_R^\top M_R=R$, $M_Q^\top M_Q=Q$.
Then $\gamma$ is an upper bound on the optimal cost of \eqref{mpc:nominal_noconstraint}.
Applying the state-feedback control $u_t=Fx_t$ with  $F=LH^{-1}$ to the system \eqref{system} leads to a cost that is guaranteed to be at most $\gamma$.
\end{mythm}
\begin{proof}
As discussed in \eqref{Vconstraint}-\eqref{worstV}, the quadratic function $V(x_t)=x_t^\top Px_t$ with $P\succ 0$ is an upper bound on the optimal cost of \eqref{mpc:nominal_noconstraint}.
Suppose $x_t^\top Px_t\leq \gamma$ holds and define $H=\gamma P^{-1}\succ0$.
Using the Schur complement, $x_t^\top Px_t\leq \gamma$ is equivalent to the inequality \eqref{sdp:nominal_no_con1}.

Additionally, $V$ is required to satisfy the inequality \eqref{Vconstraint} for any $(A, B)\in\mathcal{C}, k\in\mathbb{N}$.
We provide sufficient conditions for inequality \eqref{Vconstraint} based on LMIs.
By substituting $\bar{u}_k(t)=F\bar{x}_k(t)$, the inequality \eqref{Vconstraint} holds for all $\bar{x}_k(t), \bar{u}_k(t), k\in\mathbb{N}$ and $ (A,B)\in \mathcal{C}$ if
\begin{equation}\label{P}
    (A+BF)^\top P(A+BF)-P+F^\top RF+Q\prec 0
\end{equation}
holds for any $(A, B)\in\mathcal{C}$.
Multiplying both sides of the inequality \eqref{P} with $H=\gamma P^{-1}$, defining $L=FH$ and dividing by $\gamma$, we obtain
\begin{equation}\label{X}
    (AH+BL)^\top \!H^{-1}\!(AH+BL)\!-H+\frac{1}{\gamma}(L^\top \!R L+HQH)\!\prec\! 0.
\end{equation}
Using the Schur complement with $H\succ 0$ and defining $\Phi=\begin{bmatrix}M_R L\\ M_Q H\end{bmatrix}$, the inequality \eqref{X} is equivalent to
\begin{equation}\nonumber
    \begin{bmatrix}
        H-\frac{1}{\gamma}\Phi^\top \Phi& (AH+BL)^\top\\
        (AH+BL) &H
    \end{bmatrix}\succ 0.
\end{equation}
Using the Schur complement again, we obtain the equivalent inequalities
\begin{subequations}
\begin{align}
    &H-(AH+BL)(H\!-\frac{1}{\gamma}\Phi^\top\Phi)^{-1}\!(AH+BL)^\top\!\!\succ \!0,\label{theorem1:schur1}\\
    &H-\frac{1}{\gamma}\Phi^\top\Phi\succ 0.\label{theorem1:schur2}
\end{align}
\end{subequations}
The inequality \eqref{theorem1:schur1} is equivalent to
\begin{equation}\label{X1}
    \begin{bmatrix}
        I \\
        A^\top\\
        B^\top
    \end{bmatrix}^\top
    \!\!\!
    \begin{bmatrix}\!
        H \!&0\\
        0 \!&-\begin{bmatrix}H\\L\end{bmatrix}\!(H-\frac{1}{\gamma}\Phi^\top\Phi)^{-1}\!\begin{bmatrix}H\\L\end{bmatrix}^\top\!
    \end{bmatrix}\!\!\!
    \begin{bmatrix}
        I\\
        A^\top\\
        B^\top
    \end{bmatrix}\!\!\succ\! 0.
\end{equation}
As the set $\mathcal{C}$ is characterized by \eqref{C}, the inequality \eqref{X1} holds for any $(A, B)\in\mathcal{C}$ if there exists $\tau=(\tau_0, \ldots, \tau_{T-1})$, $\tau_i\geq 0$, $i\in\mathbb{I}_{[0, T-1]}$, such that the following inequality holds
\begin{equation}\label{M1}
    \begin{bmatrix}\!
        -H &0\\
        0 &\begin{bmatrix}H\\L\end{bmatrix}\!\!(H-\frac{1}{\gamma}\Phi^\top\Phi)^{-1}\!\!\begin{bmatrix}H\\L\end{bmatrix}^\top\!
    \end{bmatrix}+
\Pi(\tau)\prec 0.
\end{equation}
Applying again the Schur complement, \eqref{M1} together with \eqref{theorem1:schur2} is equivalent to
\begin{equation}\nonumber
    \begin{bmatrix}
        \begin{bmatrix}
            -H &0\\
            0 &0
        \end{bmatrix}+\Pi(\tau) &
        \begin{bmatrix}
            0\\
            H\\
            L
        \end{bmatrix}\\
        \begin{bmatrix}
            0 &H &L^\top
        \end{bmatrix} &-H+\frac{1}{\gamma}\Phi^\top\Phi
    \end{bmatrix}\prec 0.
\end{equation}
Using the Schur complement again, we obtain the LMI \eqref{sdp:nominal_no_con2}.
Based on the preceding discussion, if the LMIs \eqref{sdp:nominal_no_con2} and \eqref{sdp:nominal_no_con3} are satisfied, then $V$ satisfies the inequality \eqref{Vconstraint} for any $(A, B)\in\mathcal{C}$.

As we have shown earlier in \eqref{worstV}, if the function $V$ satisfies the inequality \eqref{Vconstraint} for any $(A, B)\in\mathcal{C}$, then $V(x_t)$ is an upper bound on the cost \eqref{mpc:nominal_noconstraint_obj}.
Given that \eqref{sdp:nominal_no_con1} holds, we have thus shown that $\gamma$ is an upper bound on the optimal cost of problem \eqref{mpc:nominal_noconstraint}.
\end{proof}

\begin{remark}
The proof of Theorem \ref{theorem1} bears similarities to the results presented in \cite{hguyen2023lmirobust}, where the authors aim to find a state-feedback control law that minimizes the upper bound of an infinite horizon cost functional.
However, it is important to note some key distinctions.
Specifically, in \cite{hguyen2023lmirobust}, the data used for designing the LMI-based data-driven controller are generated by the system without noise.
On the other hand, in the present scheme, the given offline data are affected by the noise satisfying instantaneous constraints.
Consequently, we need to employ different technical tools to characterize the set $\mathcal{C}$ and to derive the LMI constraints \eqref{sdp:nominal_no_con2}-\eqref{sdp:nominal_no_con3}.
\end{remark}

\begin{remark}
Solving the SDP problem \eqref{sdp:nominal_no} enables us to determine an upper bound for the objective of problem \eqref{mpc:nominal_noconstraint}.
The main source of conservatism of this upper bound is due to the linear state-feedback form of the input.
Additionally, the fact that the upper bound is assumed to be quadratic can increase the gap.
Reducing conservatism by considering more general state-feedback and cost upper bound functions is an interesting issue for future research.
\end{remark}

In the following lemma, we show that the sublevel set of the function $V(x)$ is robust positive invariant (RPI).
This result will be used to reformulate the input and state constraint in Theorem \ref{theorem2}.
\begin{mylem}\label{lemma:RPI}
Let $x_t$ be the state at time $t$.
Suppose there exist $\gamma, H, L, \tau$ such that \eqref{sdp:nominal_no_con1}-\eqref{sdp:nominal_no_con3} hold, denote the resulting state-feedback gain by $F=LH^{-1}$, and define $P=\gamma H^{-1}$.
The set $\mathcal{E}=\{x\in\mathbb{R}^n:x^\top P x\leq \gamma\}$ is an RPI set for the uncertain system $x_{t+1}=(A+BF)x_t$ with $(A, B)\in\mathcal{C}$, i.e., if
$x_t^\top P x_t\leq \gamma$,
then $x_{t+k}^\top P x_{t+k}\leq \gamma, \forall (A, B)\in\mathcal{C}, k\in\mathbb{I}_{[1, \infty)}$.
\end{mylem}
\begin{proof}
As proved in Theorem \ref{theorem1}, the inequality \eqref{Vconstraint} holds because the constraints \eqref{sdp:nominal_no_con2}-\eqref{sdp:nominal_no_con3} are satisfied.
From inequality \eqref{Vconstraint}, we have that
\begin{equation}
    \bar{x}_{k+1}(t)^\top P\bar{x}_{k+1}(t)-\bar{x}_{k}(t)^\top P\bar{x}_{k}(t)\leq -l(\bar{u}_k(t) \bar{x}_k(t))\nonumber
\end{equation}
holds for any $(A, B)\in\mathcal{C}, k\in\mathbb{N}$. Since $l(\bar{u}_k(t), \bar{x}_k(t))\geq 0$, we have
\begin{equation}
    \bar{x}_{k+1}(t)^\top P\bar{x}_{k+1}(t)\leq \bar{x}_{k}(t)^\top P\bar{x}_{k}(t).\label{lemma1:proof1}
\end{equation}
holds for any $(A, B)\in\mathcal{C}, k\in\mathbb{N}$.
Therefore, if $x_t^\top P x_t=\bar{x}_0(t)^\top P\bar{x}_0(t)\leq \gamma$, then we have \begin{equation}\nonumber
    x_{t+1}^\top Px_{t+1}\leq\max\limits_{(A, B)\in\mathcal{C}}\bar{x}_1(t)^\top P\bar{x}_1(t)\overset{\eqref{lemma1:proof1}}{\leq}\gamma.
\end{equation}
This argument can be continued for the state $x_{t+k}, \forall k\in\mathbb{I}_{[2, \infty)}$ by induction, completing the proof.
\end{proof}

Based on Lemma \ref{lemma:RPI}, the following theorem shows that the input and state constraints in \eqref{mpc:nominal_con3} and \eqref{mpc:nominal_con4} can be reformulated as LMIs.
\begin{mythm}\label{theorem2}
If \eqref{sdp:nominal_no} as well as \eqref{con:nominal} hold, then the input and state constraints \eqref{mpc:nominal_con3} and \eqref{mpc:nominal_con4} hold
\begin{subequations}\label{con:nominal}
\begin{align}
    &\begin{bmatrix}
        H & L^\top\\
        L &S_u^{-1}
    \end{bmatrix}\succeq 0,\label{con:nominal1}\\
    &\begin{bmatrix}
        S_x &I\\
        I &H
    \end{bmatrix}\succeq 0.\label{con:nominal2}
\end{align}
\end{subequations}
\end{mythm}
\begin{proof}
Suppose that a feasible solution of problem \eqref{sdp:nominal_no} is $\gamma, H, L, \tau$.
Defining $P=\gamma H^{-1}$ and $F=LH^{-1}$,
by constraint \eqref{sdp:nominal_no_con1}, we have $x_t\in\mathcal{E}=\{x\in\mathbb{R}^n:x^\top Px\leq \gamma\}$.
Given that the input is in a state-feedback form, we can write the input constraint \eqref{mpc:nominal_con3} as
\begin{equation}\nonumber
\begin{aligned}
       \max_{k\in\mathbb{I}_{\geq 0}} \|\bar{u}_k(t)\|_{S_u}^2&=\max_{k\in\mathbb{I}_{\geq 0}} \|F \bar{x}_k(t)\|_{S_u}^2\leq 1.
\end{aligned}
\end{equation}
By Lemma \ref{lemma:RPI}, the set $\mathcal{E}$ is an RPI set for the uncertain system $x_{t+1}=(A+BF)x_t$ with $(A, B)\in\mathcal{C}$.
Thus, the input constraint \eqref{mpc:nominal_con3} can be written as
\begin{equation}\label{theorem2:proof1}
\max\limits_{k\in\mathbb{I}_{\geq 0}} \|F \bar{x}_k(t)\|_{S_u}^2\leq\max\limits_{x\in\mathcal{E}}\|Fx\|_{S_u}^2\leq 1.
\end{equation}
The inequality \eqref{theorem2:proof1} holds if
\begin{equation}\nonumber
    \begin{bmatrix}
        x\\
        1
    \end{bmatrix}^\top
    \begin{bmatrix}
        -F^\top S_u F &0\\
        0 &1
    \end{bmatrix}
    \begin{bmatrix}
        x\\
        1
    \end{bmatrix}
    \geq 0,
\end{equation}
holds for all $x$ such that
\begin{equation}\nonumber
    \begin{bmatrix}
        x\\
        1
    \end{bmatrix}^\top
    \begin{bmatrix}
        -P &0\\
        0 &\gamma
    \end{bmatrix}
    \begin{bmatrix}
        x\\
        1
    \end{bmatrix}
    \geq 0.
\end{equation}
Using the S-procedure, if there exists $\lambda\geq 0$ such that
\begin{equation}\label{input:nominal_lmi1}
    \begin{bmatrix}
        -F^\top S_u F&0\\
        0&1
    \end{bmatrix}-\lambda
    \begin{bmatrix}
        -P &0\\
        0 &\gamma
    \end{bmatrix}\succeq 0,
\end{equation}
then the input constraint \eqref{mpc:nominal_con3} must be satisfied.
The inequality
\eqref{input:nominal_lmi1} holds iff the following inequalites hold
\begin{subequations}\label{input:nominal_lmi3}
    \begin{align}
        &\lambda P-F^\top S_u F\succeq 0,\label{input:nominal_lmi3a}\\
        &1-\lambda\gamma \geq 0.\label{input:nominal_lmi3b}
    \end{align}
\end{subequations}
If there exist $\gamma, P, F, \lambda$ such that \eqref{input:nominal_lmi3} hold and $\lambda<\frac{1}{\gamma}$, then there must exist $\lambda'=\frac{1}{\gamma}$ such that \eqref{input:nominal_lmi3} hold for $\gamma, P, F, \lambda'$.
Therefore, we can choose the multiplier to be $\lambda=\frac{1}{\gamma}$. Multiplying both sides of \eqref{input:nominal_lmi3a} with $H$, the inequality \eqref{input:nominal_lmi3a} is equivalent to
\begin{equation}\label{input:nominal_lmi4}
        H-L^\top S_u L\succeq 0.
\end{equation}
Using the Schur complement, the inequality \eqref{input:nominal_lmi4} is equivalent to
\eqref{con:nominal1}.
The state constraint \eqref{mpc:nominal_con4} can be written as
\begin{equation}\nonumber
\begin{aligned}
       \max_{k\in\mathbb{I}_{\geq 0}} \|\bar{x}_k(t)\|_{S_x}^2\leq \max_{x\in\mathcal{E}}\|x\|_{S_x}^2\leq 1.
\end{aligned}
\end{equation}
Thus, the state constraint \eqref{mpc:nominal_con4} holds if $x^\top S_x x\leq 1$ holds for all $x$ such that $x^\top P x\leq\gamma$.
The statement holds if $S_x\succeq \gamma^{-1} P$.
Using the Schur complement, $S_x\succeq \gamma^{-1} P$ is equivalent to \eqref{con:nominal2}.
\end{proof}

\subsection{Receding-horizon algorithm}\label{sec:3.3}
We now combine the optimization problem \eqref{sdp:nominal_no} with the constraints \eqref{con:nominal} to present the proposed data-driven min-max MPC approach.
Given an initial state $x_t$, the state-feedback gain $F_t\in\mathbb{R}^{m\times n}$ that minimizes the derived upper bound on the optimal cost of  the min-max MPC problem \eqref{mpc:nominal} can be obtained by solving the following optimization problem
\begin{subequations}\label{sdp:nominal}
\begin{align}
    &\minimize\limits_{\gamma, H, L, \tau}\gamma\label{sdp:nominal_obj}\\
    \text{s.t. }
    & \eqref{sdp:nominal_no} \text{ and } \eqref{con:nominal}\text{ hold}.
\end{align}
\end{subequations}
We denote the optimal solution of the problem \eqref{sdp:nominal} by $\gamma^\star, H^\star, L^\star, \tau^\star$.
The corresponding optimal state-feedback gain is given by $F=L^\star (H^\star)^{-1}$.

The data-driven min-max MPC problem \eqref{mpc:nominal} is solved in a receding horizon manner, see Algorithm~1.
At time $t$, we solve the optimization problem \eqref{sdp:nominal} and obtain the optimal state-feedback gain $F_t$. Only the first computed input $u_t=F_tx_t$ is implemented.
At the next sampling time $t+1$, we re-iterate the described procedure.

\begin{algorithm}[!ht]
    \SetKwData{Left}{left}
    \SetKwData{Up}{up}
    \SetKwFunction{FindCompress}{FindCompress}
    \SetKwInOut{Input}{input}
    \SetKwInOut{Output}{output}
    \nl At time $t=0$, measure state $x_0$\;
    \nl Solve the problem \eqref{sdp:nominal}\;
    \nl Apply the input $u_t=F_tx_t$\;
    \nl Set $t=t+1$ and go back to 2\;
    \label{algorithm:nominal}
    \caption{Data-driven min-max MPC.}
\end{algorithm}

\subsection{Closed-loop guarantees}\label{sec:3.4}
Now we prove recursive feasibility, constraint satisfaction and exponential stability of the closed-loop system $x_{t+1}=A_sx_t+B_su_t$ resulting from the proposed scheme.

\begin{mythm}\label{theorem3}
If the optimization problem \eqref{sdp:nominal} is feasible at time $t=0$, then
\begin{enumerate}[(i)]
\item it is feasible at any time $t\in\mathbb{N}$,
\item the closed-loop trajectory satisfies the constraints, i.e., $\|u_t\|_{S_u}\leq 1, \|x_t\|_{S_x}\leq 1$ for all $t\in\mathbb{N}$;
\item the desired equilibrium $x^s=0$ is exponentially stable for the closed-loop system.
\end{enumerate}
\end{mythm}
\begin{proof}
The proof is composed of three parts. Part I, II and III prove feasibility, constraint satisfaction and exponential stability, respectively.

\textbf{Part I:}  Let us assume that the optimization problem \eqref{sdp:nominal}  is feasible at sampling time $t$.
The only constraint in the problem that depends explicitly on the measured data $x_t$ is the inequality \eqref{sdp:nominal_no_con1}.
Constraints \eqref{sdp:nominal_no_con2}-\eqref{sdp:nominal_no_con3} and \eqref{con:nominal} are still satisfied with the feasible solution of the problem \eqref{sdp:nominal} at time $t$.
We need to prove that the inequality \eqref{sdp:nominal_no_con1} is feasible with this candidate solution for the future states $x_{t+1}$.
Suppose $\gamma, L$ and $H$ are a feasible solution of the problem \eqref{sdp:nominal} at time $t$. We define $P=\gamma H^{-1}$ and $F=LH^{-1}$.
The feasibility of the problem at time $t$ implies satisfaction of \eqref{sdp:nominal_no_con1}, which implies $x_t^\top Px_t\leq \gamma$.
Using Lemma 1, this implies that for any $(A, B)\in\mathcal{C}$, we have
\[x_t^\top (A+BF)^\top P(A+BF) x_t\leq \gamma.\]
The state at $t+1$ is $x_{t+1}=(A_s+B_sF)x_t$, where the true system matrices $(A_s, B_s)\in\mathcal{C}$.
Thus, we have
\[x_{t+1}^\top Px_{t+1}\leq \gamma.\]
Thus, the feasible solution of the optimization problem at time $t$ is also feasible at time $t+1$. This argument can be continued for time $t+2, t+3, \ldots$, which proves recursive feasibility.

\textbf{Part II:} As shown in Theorem \ref{theorem2}, if the inequalities \eqref{con:nominal} hold,
then the predicted input and state for the uncertain system $x_{t+1}=(A+BF)x_t$ with $(A, B)\in\mathcal{C}$ satisfy the input and state constraints.
Since the true system matrices $(A_s, B_s)$ lie in the set $\mathcal{C}$, the closed-loop input and state also satisfy the input and state constrains.

\textbf{Part III:} Let us denote the optimal solution of problem \eqref{sdp:nominal} at time $t$ by $H_t, \gamma_t$ and define $P_t=\gamma_t H_t^{-1}$.
Since $P_t$ is a suboptimal solution at time $t+1$, compare Part I of the proof,
we must have $x_{t+1}^\top P_{t+1}x_{t+1}\leq x_{t+1}^\top P_t x_{t+1}$.
The state at time $t+1$ is $x_{t+1}=(A_s+B_sF_t)x_t$ with $(A_s, B_s)\in\mathcal{C}$.

We choose $V(x_t)=x_{t+1}^\top P_{t+1}x_{t+1}$ as the Lyapunov function candidate.
A lower bound on $V(x_t)$ follows from the definition
$V(x_t)\geq l(u_t, x_t)\geq \|x_t\|_Q^2$.
An upper bound on $V(x_t)$ can be constructed from  $x_{t+1}^\top P_{t+1}x_{t+1}\leq x_{t+1}^\top P_t x_{t+1}$, which holds for $t\in\mathbb{N}$. This implies that $V(x_t)\leq x_t^\top P_0 x_t$.

According to the inequality \eqref{Vconstraint}, we must have
$x_{t+1}^\top P_t x_{t+1}-x_t^\top P_t x_t\leq-l(u_t, x_t)$.
Choosing $c$ as the minimum eigenvalue of $Q$, we have $l(u_t, x_t)\geq c\|x_t\|^2$.
Therefore, we can get
\[x_{t+1}^\top P_{t+1}x_{t+1}-x_t^\top P_t x_t\leq -c\|x_t\|^2.\]
This shows that $x^s=0$ is exponentially stable for the closed loop.
\end{proof}

The proof of feasibility relies on Lemma \ref{lemma:RPI}, which establishes that the closed-loop state remains contained within the set $\mathcal{E}$.
By choosing the optimal solution as a feasible solution at the next time step, we can guarantee the feasibility of problem \eqref{sdp:nominal}.
We select the Lyapunov function as $V(x)=x^\top P x$.
Using the inequality \eqref{Vconstraint}, we have that the decrease rate of the Lyapunov function is less than or equal to the negative of the stage cost.
This deduction allows us to establish exponential stability for the closed-loop system.

\begin{remark}
Existing data-driven MPC schemes based on the Fundamental Lemma \cite{coulson2019data,berberich2021guarantees} suffer from the limitation that the computational complexity increases with data length.
In contrast, the proposed approach offers larger flexibility by trading off computational tractability and conservatism. While the computational complexity of Problem \eqref{sdp:nominal} also grows with the data length via the size of $\tau$,
it is also possible to use different choices of multipliers which reduce the computational complexity at the cost of additional conservatism, see \cite[Section V.B]{berberich2023combining}. For example, when imposing the additional condition $\tau_i=\tau, \forall i \in\mathbb{I}_{[1, T]}$ for some $\tau\geq0$, the size of problem \eqref{sdp:nominal} is independent of the data length.
\end{remark}

\section{Simulation}\label{sec:4}
In this section, we apply the proposed data-driven min-max MPC scheme to the linearization of the nonlinear continuous stirred-tank reactor (CSTR) consider in \cite{mayne2011tube}.
The nonlinear system from \cite{mayne2011tube} is linearized at $\begin{bmatrix}0.9831\\0.3918\end{bmatrix}$ with a sampling time $0.5$ sec.
The linearized system is given by
\begin{equation}\label{system_numerical}
x_{t+1}=
\begin{bmatrix}
    0.9749 &-0.0135\\
    0.0004 &0.9888
\end{bmatrix}x_t+
10^{-4}\cdot\begin{bmatrix}
    0.041\\
    5.934
\end{bmatrix}u_t.
\end{equation}
The system matrices are unknown, but an input-state trajectory $(U_{-}, X)$ of length $T=200$ with additive noise $\omega_t$ is available, where the input $u\in U_{-}$ is chosen uniformly from the unit interval $[-10, 10]$.
The noise satisfies the instantaneous constraint in Assumption~\ref{assumption1}, i.e., $\omega_t\in\{\omega\in\mathbb{R}^2:\|\omega\|_2^2\leq 10^{-6}\}$.
This input-state sequence is used to characterize the set $\mathcal{C}$ in the proposed data-driven min-max MPC scheme.

In the data-driven min-max MPC problem, the input and state constraints are given by $\|u_t\|_{S_u}\leq 1$ and $\|x_t\|_{S_x}\leq 1$, where
\begin{equation}\nonumber
    S_u=0.01, S_x=\begin{bmatrix}1000 &0\\ 0 &500 \end{bmatrix}.
\end{equation}
We choose two different weighting matrices of the stage cost function to compare its influence on the closed-loop, i.e., $Q=I, R=1$ and $Q=I, R=10^{-4}$.
The initial state is given as $x_0=[-0.01,-0.04]^\top$.

Figure~\ref{simulation_result} illustrates the closed-loop input-state trajectory resulting from the application of the data-driven min-max MPC scheme.
The closed-loop trajectory converges to the origin, and the input and state constraints are satisfied in the closed-loop operation.
When choosing a smaller weighting matrix $R$, i.e., $R=10^{-4}$, the inputs obtained are larger compared to the case when $R=1$.
However, the inputs remain below the constraint bound of $10$.
This can be attributed to two factors.
First, we constrain the input to a state-feedback form. Second, we reformulate the input constraint using the S-procedure, which can introduce conservatism into the constraint.

In addition, we implement the data-driven min-max MPC scheme in the system \eqref{system_numerical} with online noise
$\omega_t\in\{\omega\in\mathbb{R}^2:\|\omega\|_2^2\leq 10^{-6}\}$, i.e., noise affecting the system dynamics during closed-loop operation.
The weight matrices are $Q=I$ and $R=10^{-4}$.
All other aspects, such as constraints and the initial state, remain consistent with the noise-free case.
The results indicate that the input and state constraints are still satisfied during the closed-loop operation. However, the closed-loop trajectory does not converge to the origin but approaches a neighborhood of the origin.
We calculate the sum of stage costs over all $300$ iterations, resulting in a value of $0.0411$, compared to a cost of $0.0369$ for the noise-free closed-loop.
This confirms that the proposed data-driven min-max MPC scheme also produces reliable results in the presence of online noise.
\begin{figure}
    \centering
    \subfigure[]{\label{state_result}
    \includegraphics[width=0.5\textwidth]{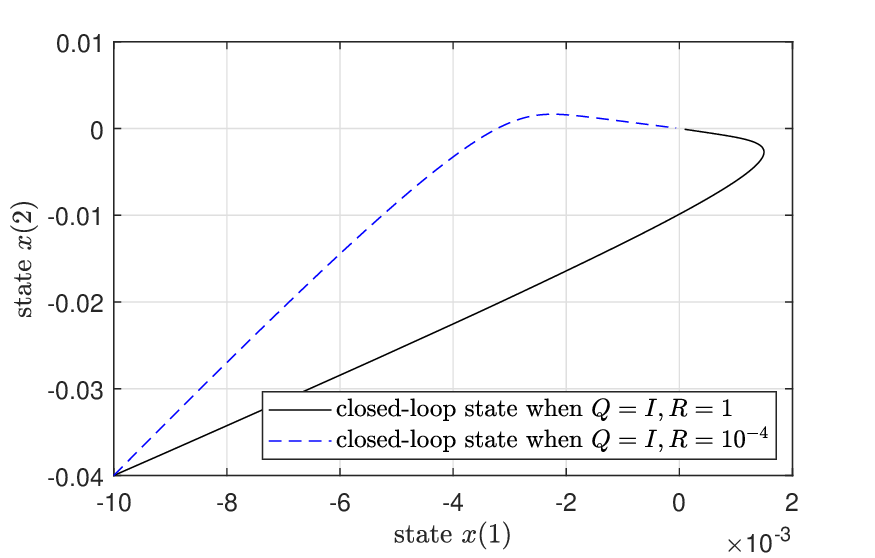}}
    \subfigure[]{\label{input_result}
    \includegraphics[width=0.5\textwidth]{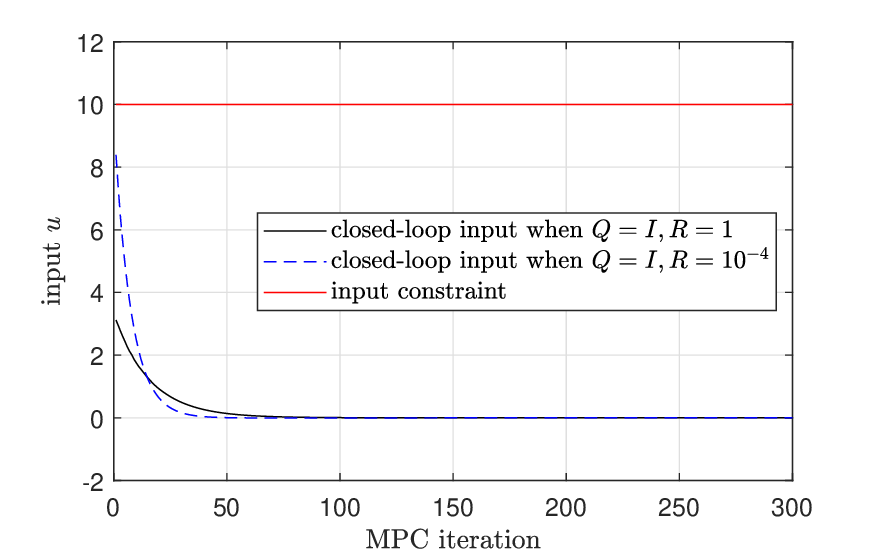}}
    \caption{Closed-loop input and state when choosing different weighting matrices. (a) Closed-loop state $x$. (b) Closed-loop input $u$.}
    \vspace{-0.7cm}
    \label{simulation_result}
\end{figure}

\section{Conclusion}\label{sec:5}
In this paper, we presented a data-driven min-max MPC scheme which uses noisy data to design controller for LTI systems with unknown system matrices.
Our main contribution is to reformulate the data-driven min-max MPC problem as an SDP and to establish that the proposed scheme guarantees closed-loop properties including recursive feasibility, constraint satisfaction and exponential stability.
A numerical example showed the effectiveness of the proposed method.
Compared with the existing min-max MPC scheme in  \cite{kothare1996robust, bemporad2003minmax, scokaert1998minmaxfeedback}, our scheme does not require any priori model knowledge.
Extending the results in this paper to robust data-driven min-max MPC for systems with online noise is an interesting future topic.

\bibliographystyle{plain}
\bibliography{main}
	
\end{document}